\documentclass[a4paper,UKenglish]{lipics-v2018}

\usepackage{microtype}

\usepackage{amsmath}
\usepackage{amsfonts}
\usepackage{amssymb, amsthm, epsfig}
\usepackage{mathrsfs}

\usepackage{graphicx}
\usepackage{verbatim}
\usepackage{color}
\usepackage{epstopdf}

\usepackage{float}
\usepackage{algorithm}
\usepackage[noend]{algorithmic}

\newcommand{\s}[1]{{\textsf{#1}}}
\newcommand{\cost}{\ensuremath{\s{cost}}}
\newcommand{\Eu}[1]{\ensuremath{\EuScript{#1}}}
\newcommand{\eps}{\varepsilon}
\newcommand{\etal}{\emph{et{.}al{.}}}
\newcommand{\bb}[1]{\ensuremath{\mathbb{#1}}}
\newcommand{\all}{\ensuremath{\s{flat}}}%

\hideLIPIcs

\title{Approximating the Distribution of the Median and other Robust Estimators on Uncertain Data }

\titlerunning{Approximating the Distribution of the Median on Uncertain Data}

\author{Kevin Buchin}{Department of Mathematics and Computer Science, TU Eindhoven\\{[Eindhoven, The Netherlands]}}{k.a.buchin@tue.nl}{https://orcid.org/0000-0002-3022-7877}{}

\author{Jeff M. Phillips}{School of Computing, University of Utah\\{[Salt Lake City, USA]}}{jeffp@cs.utah.edu}{}{}

\author{Pingfan Tang}{School of Computing, University of Utah\\{[Salt Lake City, USA]}}{tang1984@cs.utah.edu}{}{}

\authorrunning{K. Buchin, J.\,M. Phillips, and P. Tang}

\Copyright{K. Buchin, J.\,M. Phillips, and P. Tang}

\subjclass{Theory of computation $\sim$ Computational geometry}


\category{}
\funding{NSF CCF-1115677, CCF-1350888, IIS-1251019, ACI-1443046, CNS-1514520, CNS-1564287, and NWO project no.~612.001.207}

\EventEditors{Bettina Speckmann and Csaba D. T{\'o}th}
\EventNoEds{2}
\EventLongTitle{34th International Symposium on Computational Geometry (SoCG 2018)}
\EventShortTitle{SoCG 2018}
\EventAcronym{SoCG}
\EventYear{2018}
\EventDate{June 11--14, 2018}
\EventLocation{Budapest, Hungary}
\EventLogo{socg-logo}
\SeriesVolume{99}
\ArticleNo{}

\nolinenumbers
\begin{document}

\maketitle

\begin{abstract}
Robust estimators, like the median of a point set, are important for data analysis in the presence of outliers. We study robust estimators for locationally uncertain points with discrete distributions. That is, each point in a data set has a discrete probability distribution describing its location. The probabilistic nature of uncertain data makes it challenging to compute such estimators, since the true value of the estimator is now described by a distribution rather than a single point.
We show how to construct and estimate the distribution of the median of a point set.  Building the approximate support of the distribution takes near-linear time, and assigning probability to that support takes quadratic time.
We also develop a general approximation technique for distributions of robust estimators with respect to ranges with bounded VC dimension.  This includes the geometric median for high dimensions and the Siegel estimator for linear regression.
 \end{abstract}

\section{Introduction}
\label{intro}
\vspace{-.1in}

Most statistical or machine learning models of noisy data start with the assumption that a data set is drawn iid (independent and identically distributed) from a single distribution.  Such distributions often represent some true phenomenon under some noisy observation. Therefore, approaches that mitigate the influence of noise, involving robust statistics or regularization, have become commonplace.

However, many modern data sets are clearly not generated iid, rather each data element represents a separate object or a region of a more complex phenomenon.  For instance, each data element may represent a distinct person in a population or an hourly temperature reading.  Yet, this data can still be noisy; for instance, multiple GPS locational estimates of a person, or multiple temperature sensors in a city.  The set of data elements may be noisy \emph{and} there may be multiple inconsistent readings of each element.  To model this noise, the inconsistent readings can naturally be interpreted as a probability distribution.

Given such locationally noisy, non-iid data sets, there are many unresolved and important analysis tasks ranging from classification to regression to summarization.
In this paper, we initiate the study of robust estimators~\cite{DH83,Rou85} on locationally uncertain data.  More precisely, we consider an input data set of size $n$, where each data point's location is described by a discrete probability distribution.  We will assume these discrete distributions have a support of at most $k$ points in $\bb{R}^d$; and for concreteness and simplicity we will focus on cases where each point has support described by exactly $k$ points, each being equally likely.

Although algorithms for locationally uncertain points have been studied in quite a few contexts over the last decade~\cite{CM08,LP09,JLP11,AHSYZ14,HL12,AESZ12,AAHPYZ13,ACTY09,ZLTZ09}
(see more through discussion in full version~\cite{PT16}),
few have directly addressed the problem of noise in the data.  As the uncertainty is often the direct consequence of noise in the data collection process, this is a pressing concern.  As such we initiate this study focusing on the most basic robust estimators: the median for data in $\bb{R}^1$, as well as its generalization the geometric median and the Tukey median for data in $\bb{R}^d$, defined in Section \ref{sec:prelim}.
Being robust refers to the fact that the median and geometric medians have a \emph{breakdown point}s of 0.5, that is, if less than $50\%$ of the data points (the outliers) are moved from the true distribution to some location infinitely far away, the estimator remains within the extent of the true distribution~\cite{RL91}.
The Tukey median has a breakdown point between $\frac{1}{d+1}$ and $\frac{1}{3}$ ~\cite{Alo06}.

In this paper, we generalize the median (and other robust estimators) to locationally uncertain data, where the outliers can occur not just among the $n$ data points, but also as part of the discrete distributions representing their possible locations.

The main challenge is in modeling these robust estimators.  As we do not have precise locations of the data, there is not a single minimizer of $\cost(x,Q)$; rather there may be as many as $k^n$ possible input point sets $Q$ (the combination of all possible locations of the data).  And the expected value of such a minimizer is not robust in the same way that the mean is not robust.  As such we build a distribution over the possible locations of these cost-minimizers.  In $\bb{R}^1$ (by defining boundary cases carefully) this distribution is of size at most $O(nk)$, the size of the input, but already in $\bb{R}^2$ it may be as large as $k^n$.

\vspace{-.1in}
\subparagraph*{Our Results.}
We design algorithms to create an approximate support of these median distributions.  We create small sets $T$ (called an \emph{$\eps$-support}) such that each possible median $m_Q$ from a possible point set $Q$ is within a distance $\eps \cdot \cost(m_Q, Q)$ of some $x \in T$. In $\bb{R}$ we can create a support set $T$ of size $O(k/\eps)$  in $O(nk \log (nk))$ time. We show that the bound $O(k/\eps)$ is tight since there may be $k$ large enough modes of these distributions, each requiring $\Omega(1/\eps)$ points to represent.
In $\bb{R}^d$ our bound on $|T|$ is $O(k^d/\eps^d)$, for the Tukey median and the geometric median.
If we do not need to cover sets of medians $m_Q$ which occur with probability less than $\eps$, we can get a bound $O(d/\eps^2)$ in $\bb{R}^d$.
In fact, this general approach in $\bb{R}^d$ extends to other estimators, including the Siegel estimator~\cite{Sie82} for linear regression.
We then need to map weights onto this support set $T$. We can do so exactly in $O(n^2 k)$ time in $\bb{R}^1$ or approximately in $O(1/\eps^2)$ time in $\bb{R}^d$.

Another goal may be to then construct a single-point estimator of these distributions: the median of these median distributions.  In $\bb{R}^1$ we can show that this process is stable up to $\cost(m_Q, Q)$ where $m_Q$ is the resulting single-point estimate.  However, we also show that already in $\bb{R}^1$ such estimators are not stable with respect to the weights in the median distribution, and hence not stable with respect to the probability of any possible location of an uncertain point.  That is, infinitesimal  changes to such probabilities can greatly change the location of the single-point estimator.  
As such, we argue the approximate median distribution (which is stable with respect to these changes) is the best robust representation of such data.

\subsection{Formalization of Model and Notation}
\label{sec:prelim}
We consider a set of $n$ locationally uncertain points $\Eu{P} = \{P_1, \ldots, P_n\}$ so that each $P_i$ has $k$ possible locations $\{p_{i,1}, \ldots, p_{i,k}\} \subset \mathbb{R}^d$.
Here, $P_i=\{p_{i,1}, \ldots, p_{i,k}\}$ is a multiset, which means a point in $P_i$ may appear more than once.
Let $P_\all = \cup_i \{p_{i,1}, \ldots, p_{i,k}\}$ represent all positions of all points in $\Eu{P}$, which implies $P_\all$ is also a multiset.
We consider each $p_{i,j}$ to be an equally likely (with probability $1/k$) location of $P_i$, and can extend our techniques to non-uniform probabilities and uncertain points with fewer than $k$ possible locations.
For an uncertain point set $\Eu{P}$ we say $Q \Subset \Eu{P}$ is a \emph{traversal} of $\Eu{P}$ if $Q = \{q_1, \ldots q_n\}$ has each $q_i$ in the domain of $P_i$ (e.g., $q_i = p_{i,j}$ for some $j$).  We denote by $\Pr_{Q \Subset \Eu{P}}[\gamma(Q)]$ the probability of the event $\gamma(Q)$, given that $Q$ is a randomly selected traversal from $\Eu{P}$, where the selection of each $q_i$ from $P_i$ is independent of $q_{i'}$ from $P_{i'}$.

We are particularly interested in the case where $n$ is large and $k$ is small.  
For technical simplicity we assume an extended RAM model where $k^n$ (the number of possible traversals of point sets) can be computed in $O(1)$ time and fits in $O(1)$ words of space.

We consider three definitions of medians.
In one dimension, given a set $Q = \{q_1, q_2, \ldots, q_n\}$ that w.l.o.g.\ satisfies $q_1 \leq q_2 \leq \ldots \leq q_n$, we define the \emph{median} $m_Q$ as $q_{\frac{n+1}{2}}$ when $n$ is odd and $q_{\frac{n}{2}}$ when $n$ is even.
There are several ways to generalize the median to higher dimensions~\cite{Alo06}, herein we focus on the geometric median and Tukey median.
Define $\cost(x,Q) = \frac{1}{n} \sum_{i=1}^n \|x - q_i\|$ where $\|\cdot\|$ is the Euclidian norm.  Given a set $Q = \{q_1, q_2, \ldots, q_n\} \subset \bb{R}^d$, the \emph{geometric median} is defined as $m_Q = \arg \min_{x \in \bb{R}^d} \cost(x,Q)$.
The Tukey depth~\cite{JWT1975} of a point $p$ with respect to a set $Q \subset \bb{R}^d$ is defined
$\mathsf{depth}_Q(p) :=\min_{H\in \mathcal{H}_p} |H\cap Q|$
where
$\mathcal{H}_p :=\{H \text{ is a closed half space in } \bb{R}^d \mid  p\in H\}$.
Then a \emph{Tukey median} of a set $Q$ is a point $p$ that can maximize the Tukey depth.

\subsection{Related Work on Uncertain Data}
\label{sec:related}
The algorithms and computational geometry communities have recently generated a large amount of research in trying to understand how to efficiently process and represent uncertain data~\cite{CM08,LP09,JLP11,AHSYZ14,HL12,KCS11a,AESZ12,AAHPYZ13,ACTY09,ADP13}, not to mention some motivating systems and other progress from the database community~\cite{1644250,drs-pddd-09,efficientquery,ZLTZ09,ABSHNSW06}.
Some work in this area considers other models, with either worst-case representations of the data uncertainty~\cite{kl-lbbsd-10} which do not naturally allow probabilistic models, or when the data may not exist with some probability~\cite{HL12,KCS11a,AHSYZ14}.  The second model can often be handled as a special case of the locationally uncertain model we study.
Among locationally uncertain data, most work focuses on data structures for easy data access~\cite{threshquery,efficientquery,TCXNKP05,ACTY09} but not the direct analysis of data.
Among the work on analysis and summarization, such as for histograms~\cite{CG09}, convex hulls~\cite{AHSYZ14}, or clustering~\cite{CM08} it usually focuses on quantities like the expected or most likely value, which may not be stable with respect to noise.  This includes estimation of the expected median in a stream of uncertain data~\cite{JMMV08} or the expected geometric median as part of $k$-median clustering of uncertain data~\cite{CM08}.
We are not aware of any work on modeling the probabilistic nature of locationally uncertain data to construct robust estimators of that data, robust to outliers in both the set of uncertain points as well as the probability distribution of each uncertain point.


\section{Constructing a Single Point Estimate}
\label{sec:single-pt}

We begin by exploring the construction of a single point estimator of set of $n$ locationally uncertain points $\Eu{P}$.  We demonstrate that while the estimator is stable with respect to the value of $\cost$, the actual minimum of that function is not stable and provides an incomplete picture for multimodal uncertainties.

It is easiest to explore this through a weighted point set $X \subset \bb{R}^1$.  Given a probability distribution defined by $\omega : X \to [0,1]$, we can compute its weighted median by scanning from smallest to largest until the sum of weights reaches $0.5$.

There are two situations whereby we obtain such a discrete weighted domain.
The first domain is the set $T$ of possible locations of medians under different instantiations of the uncertain points with weights $\hat w$ as the probability of those medians being realized; see constructions in Section \ref{Building the eps-Support T and Bounding its Size} and Section \ref{sec:weight-T}.  Let the resulting weighted median of $(T, \hat w)$ be $m_T$.
The second domain is simply the set $P_\all$ of all possible locations of $\Eu{P}$, and its weight $w$ where $w(p_{i,j})$ is the fraction of $Q \Subset \Eu{P}$ which take $p_{i,j}$ as their median (possibly $0$).  Let the weighted median of $(P_\all, w)$ be $m_{\Eu{P}}$.

\begin{theorem} \label{thm:single-pnt-err}
$|m_T - m_\Eu{P}| \leq \eps \cost(m_\Eu{P}) \leq \eps \cost(m_Q,Q)$, $Q \Subset \Eu{P}$ is any traversal with $m_\Eu{P}$ as its median.
\end{theorem}
\begin{proof}
We can divide $\bb{R}$ into $|T|$ intervals, one associated with each $x \in T$, as follows.  Each $z \in \bb{R}$ is in an interval associated with $x \in T$ if $z$ is closer to $x$ than any other point $y \in T$, unless $|z - y| \leq \varepsilon \cost(z)$ but $|z - x| > \cost(z)$.  Thus a point $p_{i,j}$ whose weight $w(p_{i,j})$ contributes to $\hat w(x)$, is in the interval associated with $x$.

Thus, if $p_{i,j} = m_\Eu{P}$, then the sum of all weights of all points greater than $p_{i,j}$ is at most $0.5$, and the sum of all weights of points less than $p_{i,j}$ is less than $0.5$.  Hence if $m_\Eu{P}$ is in an interval associated with $x \in T$, then the sum of all weights of points $p_{i,j}$ in intervals greater than that of $x$ must be at most $0.5$ and those less than that of $x$ must be less than $0.5$.  Hence $m_T = x$, and $|x - p_{i,j}| \leq \eps \cost(m_\Eu{P})$ as desired.
\end{proof}

\vspace{-.1in}
\subparagraph*{Non-Robustness of single point estimates.}
The geometric median of the set $\{m_Q \text{ is a}$ geometric median of $Q \mid Q \Subset \Eu{P}\}$ is not stable under small perturbations in weights; it stays within the convex hull of the set, but otherwise not much can be said, even in $\bb{R}^1$.  Consider the example with $n=3$ and $k=2$, where $p_{1,1} = p_{1,2} = p_{2,1} = 0$ and $p_{2,2} = p_{3,1} = p_{3,2} = \Delta$ for some arbitrary $\Delta$.  The median will be at $0$ or $\Delta$, each with probability $1/2$, depending on the location of $P_2$.
We can also create a more intricate example where $\hat \cost (0) = \hat \cost(\Delta) = 0$.
As these examples have $m_Q$ at $0$ or $\Delta$ equally likely with probability $1/2$, then canonically in $\bb{R}^1$ we would have the median of this distribution at $0$, but a slight change in probability (say from sampling) could put it all the way at $\Delta$.
This indicates that a representation of the distribution of medians as we study in the remainder is more appropriate for noisy data.

\section{Approximating the Median Distribution}
\label{sec:T}

%

The big challenge in constructing an $\eps$-support $T$ is finding the points $x \in P_\all$ which have small values of $\cost(x, Q)$ (recall $\cost(x,Q) = \frac{1}{n} \sum_{i=1}^n \|x - q_i\|$) for some $Q \Subset \Eu{P}$.  But this requires determining the smallest cost $Q \Subset \Eu{P}$ that has $x \in Q$ and $x$ is the median of $Q$.

One may think (as the authors initially did) that one could simply use a proxy function
$\hat \cost(x) = \frac{1}{n} \sum_{i=1}^n \min_{1 \leq j \leq k} \|x - p_{i,j}\|$, which is relatively simple to compute as the lower envelope of cost functions for each $P_i$.  Clearly $\hat \cost(x) \leq \cost(x,Q)$ for all $Q \Subset \Eu{P}$, so a set $\hat T$ satisfying a similar approximation for $\hat \cost$ will satisfy our goals for $\cost$.  However, there exist (rather adversarial) data sets $\Eu{P}$ where $\hat T$ would require $\Omega(nk)$ points;
see Appendix \ref{the size of hat{T}}.
On the other hand, we show this is not true for $\cost$.
The key difference between $\cost$ and $\hat \cost$ is that $\hat \cost$ does not enforce the use of some $Q \Subset \Eu{P}$ of which $x$ is a median.  That is, that (roughly) half the points are to the left and half to the right for this $Q$.

\vspace{-.1in}
\subparagraph*{Proxy functions $L$, $R$, and $D$.}
We handle this problem by first introducing two families of functions, defined precisely shortly.  We let $L_i(x)$ (resp.\ $R_i(x)$) represent the contribution to $\cost$ at $x$ from the closest possible location $p_{i,j}$ of an uncertain point $P_i$ to the left (resp.\ right) of $x$.  This allows us to decompose the elements of this cost.  However, it does not help us to enforce this balance.  Hence we introduce a third proxy function
\[
D_i(x) = L_i(x) - R_i(x)
\]
capturing the difference between $L_i$ and $R_i$.  We will show that the choice of which points are used on the left or right of $x$ is completely determined by the $D_i$ values.  In particular, we maintain the $D_i$ values (for all $i \in [n]$) in sorted order, and use the $i$ with larger $D_i$ values on the right, and smaller $D_i$ values on the left for the min cost $Q \Subset \Eu{P}$.


To define $L_i$, $R_i$, and $D_i$, we first assume that $P_\all$ and $P_i$  for all $i\in[n]$ are sorted (this would take $O(nk \log (nk))$ time).  Then to simplify definitions we add                                                                                    two dummy points to each $P_i$, and introduce the notation
$\widetilde{P}_i=P_i\cup\{p_{i,0},p_{i,k+1}\}$ and
$\widetilde{\Eu{P}}=\{\widetilde{P}_1, \widetilde{P}_2,\cdots, \widetilde{P}_n\}$,
where $p_{i,0}=\min P_\all-n \Delta$, $p_{i,k+1}=\max P_\all +n\Delta$, and $\Delta=\max P_\all-\min P_\all$.
Thus, every point $p\in P_\all$ can be viewed as the median of some traversal of $\widetilde{\Eu{P}}$.  Moreover, since we put the $p_{i,0}$ and $p_{i,k+1}$ points far enough out, they will essentially act as points at infinity and not affect the rest of our analysis.

Next, for $p\in P_\all$ we define
$\cost(p)=\min \{\cost(p,Q)\ |\ p \text{ is the median of $Q$ and } Q \Subset \widetilde{\Eu{P}} \}$.
Thus, if there exists $Q \Subset \Eu{P}$ such that $p$ is the median of $Q$, then $\cost(p)\leq \cost(p,Q)$.

Now to compute $\cost$ and expedite our analysis, for $p\in[\min P_\all-n\Delta, \max P_\all+n\Delta]$, we define
$L_i(p)=\min\{|p_i-p| \ | \ p_i\in \widetilde{P}_i \cap (-\infty, p] \}$ and
$R_i(p)=\min\{|p_i-p| \ | \ p_i\in \widetilde{P}_i \cap [p,\infty) \}.$
and recall $D_i(p) = L_i(p) - R_i(p)$. Obviously, if $p\in \widetilde{P}_i$, then $D_i(p)=L_i(p)=R_i(p)=0$.
For example, if $\widetilde{P}_i=\{p_{i,0},\ p_{i,1},\ p_{i,2},\ p_{i,3},\ p_{i,4}\}$ and $p_{i,0}<p_{i,1}<p_{i,2}<p_{i,3}<p_{i,4}$, then the plot of $L_i(p)$, $R_i(p)$ and $D_i(p)$, is shown in Figure \ref{fig0}.

\begin{figure}[t]
  \vspace{-.2in}
  \includegraphics[width=0.8\textwidth]{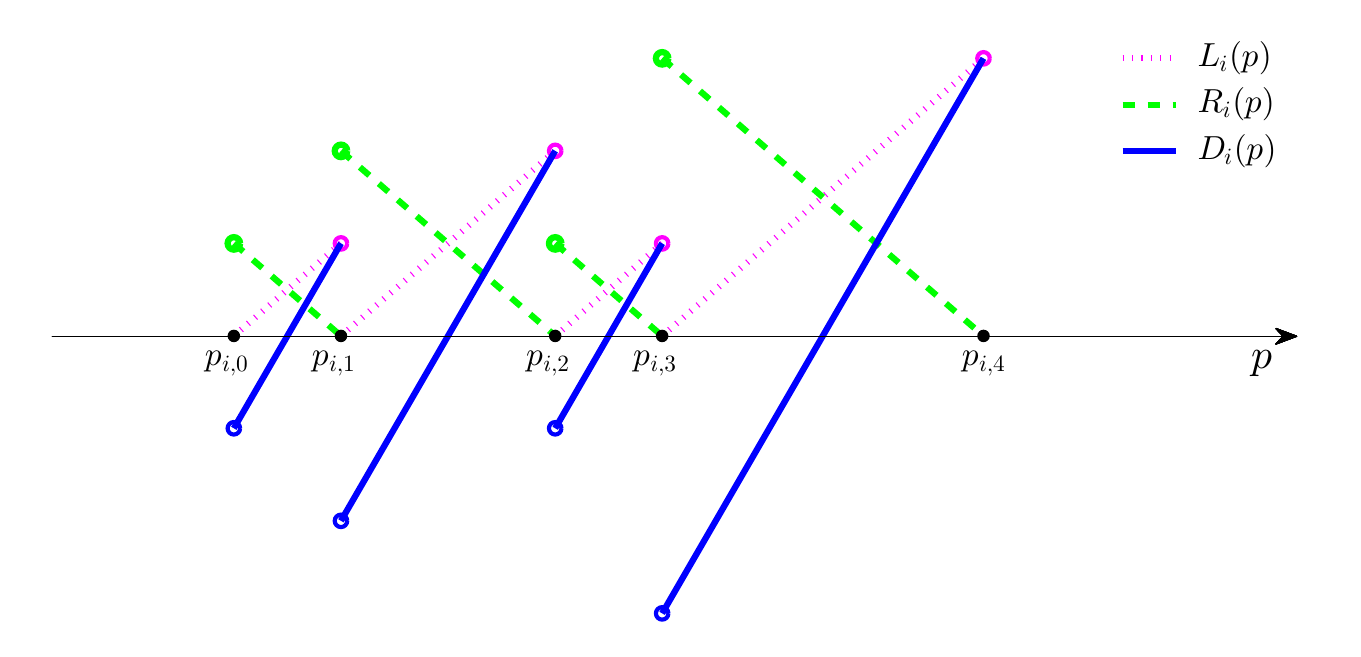}\\

  \vspace{-0.3in}
  \caption{The plot of $L_i(p)$, $R_i(p)$ and $D_i(p)$.}\label{fig0}
  \vspace{-.2in}
\end{figure}

For the sake of brevity, we now assume $n$ is odd; adjusting a few arguments by $+1$ will adjust for the $n$ is even case.

Consider next the following property of the $D_i$ functions with respect to computing $\cost(p)$ for a point $p \in P_{i_0}$.  Let $\{i_1,i_2,\cdots,i_{n-1}\}=[n]\backslash \{i_0\}$ be a permutation of uncertain points, except for $i_0$, so that $D_{i_1}(p)\leq D_{i_2}(p)\leq \cdots \leq D_{i_{n-1}}(p)$.  Then to minimize $\cost(p, Q)$, we count the uncertain points $P_{i_l}$ using $L_{i_l}$ if in the permutation $i_l \leq (n-1)/2$ and otherwise count it on the right with $R_{i_l}$.  This holds since for any other permutation $\{j_1,j_2,\cdots,j_{n-1}\}=[n]\backslash \{i_0\}$ we have
$
\sum_{l=\frac{n+1}{2}}^{n-1} D_{i_l}(p)
\geq
\sum_{l=\frac{n+1}{2}}^{n-1} D_{j_l}(p)
$
and thus
\begin{align*}
\sum_{l=1}^{\frac{n-1}{2}}L_{i_l}(p)+\sum_{l=\frac{n+1}{2}}^{n-1} R_{i_l}(p)
 &=
\sum_{l=1}^{n-1} L_{i_l}(p)-\sum_{l=\frac{n+1}{2}}^{n-1} D_{i_l}(p)
\\ &\leq
\sum_{l=1}^{n-1} L_{j_l}(p)-\sum_{l=\frac{n+1}{2}}^{n-1} D_{j_l}(p)
=
\sum_{l=1}^{\frac{n-1}{2}}L_{j_l}(p)+\sum_{l=\frac{n+1}{2}}^{n-1} R_{j_l}(p).
\end{align*}
For $p \in P_{i_0}$, $\cost(p) = \frac{1}{n}\big(\sum_{l=1}^{\frac{n-1}{2}}L_{i_l}(p)+\sum_{l=\frac{n+1}{2}}^{n-1} R_{i_l}(p)\big)$ under this $D_i$-sorted permutation.

\subsection{Computing cost}
Now to compute $\cost$ for all points $p \in P_\all$, we simply need to maintain the $D_i$ in sorted order, and then sum the appropriate terms from $L_i$ and $R_i$.  Let us first examine a few facts about the complexity of these functions.

The function $L_i$ (resp.\  $R_i$) is piecewise-linear, where the slope is always $1$ (resp.\ $-1$).  The breakpoints only occur at $x = p_{i,j}$ for each $p_{i,j} \in P_i$. Hence, they each have complexity $\Theta(k)$ for all $i \in [n]$.
The structure of $L_i$ and $R_i$ implies that $D_i$ is also piecewise-linear, where the slope is always $2$ and has breakpoints for each $p_{i,j} \in P_i$.  Each linear component attains a value $D_i(x) = 0$ when $x$ is the midpoint between two $p_{i,j}, p_{i,j'} \in P_i$ which are consecutive in the sorted order of $P_i$.

The fact that all $D_i$ have slope $2$ at all non-discontinuous points, and these discontinuous points only occur at $P_i$, implies that the sorted order of the $D_i$ functions does not change in between points of $P_\all$.  Moreover, at one of these points of discontinuity $x \in P_\all$, the ordering between $D_i$s only changes for uncertain points $D_{i'}$ such that there exists a possible location $p_{i',j} \in P_{i'}$ such that $x = p_{i',j}$.  This implies that to maintain the sorted order of $D_i$ for any $x$, as we increase the value of $x$, we only need to update this order at the $nk$ points in $P_\all$ with respect to $D_{i'}$ for which there exists $p_{i',j} \in P_{i'}$ with $p_{i',j} = x$.   This takes $O(\log (nk))$ time per update using a balanced BST, and thus $O(nk \log (nk))$ time to define $\cost(x)$ for all values $x \in \bb{R}^1$.   To compute $\cost(x)$, we also require the values of $L_i$ (or $R_i$); these can be constructed independently for each $i \in [n]$ in $O(k)$ time after sorting, and in $O(nk \log k)$ time overall.%
\footnote{
When multiple distinct $p_{i,j}$ coincide at a point $x$, then more care may be required to compute $\cost(x)$ (depending on the specifics of how the median is defined in these boundary cases).  Specifically, we may not want to set $L_i(x) = 0$, instead it may be better to use the value $R_i(x)$ even if $R_i(x) = \alpha > 0$.  This is the case when $\alpha < R_{i'}(x) - L_{i'}(x)$ for some other uncertain point $P_{i'}$ (then we say $P_i$ is on the right, and $P_{i}$ is on the left).  This can be resolved by either tweaking the definition of median for these cases, or sorting all $D_i(x)$ for uncertain points $P_i$ with some $p_{i,j} = x$, and some bookkeeping.
}
Ultimately, we arrive at the following theorem.
\begin{theorem} \label{thm:comp-cost}
Consider a set of $n$ uncertain points $\Eu{P}$ with $k$ possible locations each.  We can compute $\cost(x)$ for all $x \in \bb{R}$ such that $x = p_{i,j}$ for some $p_{i,j} \in P_\all$ in $O(nk \log (nk))$ time.
\end{theorem}

\subsection{Building the $\eps$-Support $T$ and Bounding its Size}
\label{Building the eps-Support T and Bounding its Size}

We next show that there always exists an $\eps$-support $T$ and it has a size $|T|=O(\frac{k}{\eps})$.

\begin{theorem}
\label{theorem the size of T}
Given a set of $n$ uncertain points $\Eu{P}=\{P_1,\cdots,P_n\}$, where $P_i=\{p_{i,1},\cdots,p_{i,k}\}$ $\subset\mathbb{R}$,
 and $\eps \in(0,1]$ we can construct an $\eps$-support $T$ that has a size $|T|=O(\frac{k}{\eps})$.
\end{theorem}

\begin{proof}
We first sort $P_\all$ in ascending order, scan $P_\all=\{p_1,\cdots,p_{nk}\}$ from left to right
and choose one point from $P_\all$ every $\lfloor \frac{n}{3}\rfloor$ points, and then put the chosen point into $T$.
Now, suppose $p$ is the median of some traversal $Q\Subset \Eu{P}$ and $\cost(p)=\cost(p,Q)$. If $p\notin T$, then there
are two consecutive points $t,t'$ in $T$ such that $t<p<t'$. On either side of $p$ there are at least $\lfloor \frac{n}{2} \rfloor$ points in $Q$, so without loss of generality, we assume $|p-t'|\geq \frac{1}{2}|t-t'|$.  Since $|[p,\infty)\cap Q| \geq \frac{n}{2}$ and there are at most $\lfloor \frac{n}{3}\rfloor$ points in $[p,t']$, we have $|(t',\infty)\cap Q| \geq \frac{n}{2}-\lfloor \frac{n}{3}\rfloor\geq \frac{n}{6}$, which implies
\begin{equation} \label{the lower bound of cost(p)}
\begin{split}
\cost(p)=&\cost(p,Q)\geq \frac{1}{n}\sum_{q\in (t',\infty)\cap Q}|q-p|\geq\frac{1}{n}\sum_{q\in (t',\infty)\cap Q}|t'-p|\\
\geq&\frac{1}{n}\frac{n}{6}|t'-p|=\frac{1}{6}|t'-p|\geq \frac{1}{12}|t-t'|.
\end{split}
\end{equation}
For any fixed $\varepsilon \in (0,1]$, and two consecutive points $t,t'$ $(t<t')$ in $T$, we put $x_1,\cdots,x_{\lceil \frac{12}{\varepsilon}\rceil-1}$ into $T$ where $x_i=t+\frac{|t-t'|i}{\lceil \frac{12}{\varepsilon}\rceil}$ for $1\leq i \leq \lceil \frac{12}{\varepsilon}\rceil-1$. So, for the median $p \in (t,t')$, there exists $x_{i}\in T$ s.t. $|p-x_i| \leq \frac{\varepsilon}{12}|t-t'|$, and from \eqref{the lower bound of cost(p)}, we know $|p-x_i|\leq \varepsilon \cost(p)$. In total we put $O(\frac{k}{\varepsilon})$ points into $T$; thus the proof is completed.
\end{proof}

\theoremstyle{remark}
\newtheorem{myremark}{Remark}

\begin{myremark}
The above construction results in an $\eps$-support $T$ of size $O(k/\eps)$, but does not restrict that $T \subset P_\all$.  We can enforce this restriction by for each $x$ placed in $T$ to choose the single nearest point $p \in P_\all$ to replace it in $T$.  This results in an $(2\eps)$-support, which can be made an $\eps$-support by instead adding $\lceil \frac{24}{\eps} \rceil-1$ points between each pair $(t,t')$, without affecting the asymptotic time bound.
\end{myremark}

\begin{myremark}
We can construct a sequence of uncertain data $\{\Eu{P}(n,k)\}$ such that, for each uncertain data $\Eu{P}(n,k)$, the optimal $\eps$-support $T$ has a size $\Omega(\frac{k}{\eps})$. For example, for $\eps=\frac{1}{3}, \frac{1}{5},\frac{1}{7}, \cdots$, we
define $n=\frac{1}{\eps}$, and $p_{i,j}=(j-1)n+i$ for $i\in[n]$ and $j\in[k]$.
Then, for any median $p\in P_\all$,
we have $\eps\cost(p)=\frac{2}{n^2}\sum_{i=1}^{\frac{n-1}{2}}i=\frac{n^2-1}{4n^2}<\frac{1}{4}$, hence covering no other points,
which implies $|T|=\Omega(nk)=\Omega(\frac{k}{\eps})$.
\end{myremark}

We can construct the minimal size $\eps$-support $T$ in $O(nk \log (nk))$ time by sorting, and greedily adding the smallest point not yet covered each step.  This yields the slightly stronger corollary  of Theorem \ref{theorem the size of T}.

\begin{corollary}
Consider a set of $n$ uncertain points $\Eu{P}=\{P_1,\cdots,P_n\}$, where $P_i=\{p_{i,1},\cdots,p_{i,k}\}\subset \mathbb{R}$, and $\eps \in(0,1]$.  We can construct an $\eps$-support $T$ in $O(nk \log (nk))$ time which has the minimal size for any $\eps$-support, and $|T| = O(\frac{k}{\eps})$.
\end{corollary}



There are multiple ways to generalize the notion of a median to higher dimensions~\cite{Alo06}.  We focus on two variants: the Tukey median and the geometric median.
%
We start with generalizing the notion of an $\eps$-support to a Tukey median since it more directly follows from the techniques in Theorem \ref{theorem the size of T}, and then address the geometric median.

\subsection{An $\eps$-Support for the Tukey Median}
\label{sec:Tukey}
A closely related concept to the Tukey median is a \emph{centerpoint}, which is a point $p$ such that $\mathsf{depth}_Q(p)\geq \frac{1}{d+1}|Q|$. Since for any finite set $Q\in \mathbb{R}^d$ its centerpoint always exists, a Tukey median must be a centerpoint.
This means if $p$ is the Tukey median of $Q$, then for any closed half space containing $p$, it contains at least $\frac{1}{d+1}|Q|$ points of $Q$.
Using this property, we can prove the following theorem.

\begin{theorem}
\label{theorem the epsilon support for Tukey median}
Given a set of $n$ uncertain points $\Eu{P}=\{P_1,\cdots,P_n\}$, where $P_i=\{p_{i,1},\cdots,p_{i,k}\}$ $\subset\mathbb{R}^2$,
 and $\eps \in(0,1]$, we can construct an $\eps$-support $T$ for the Tukey median on $\Eu{P}$ that has a size $|T|=O(\frac{k^2}{\eps^2})$.
\end{theorem}

\begin{proof}
Suppose the projections of $P_\all$ on $x$-axis and $y$-axis are $X$ and $Y$ respectively. We sort all points in $X$ and
choose one point from $X$ every $\lfloor \frac{n}{4}\rfloor$ points, and then put the chosen points into a set $X_T$.
For each point $x\in X_T$ we draw a line through $(x,0)$ parallel to $y$-axis.
Similarly, we sort all points in $Y$ and
choose one point every $\lfloor \frac{n}{4}\rfloor$ points, and put the chosen points into $Y_T$.
For each point $y\in Y_T$ we draw a line through $(0,y)$ parallel to $x$-axis.

Now, suppose $p$ with coordinates $(x_p,y_p)$ is the Tukey median of some traversal $Q\Subset \Eu{P}$ and
$\cost(p,Q)=\frac{1}{n}\sum_{q\in Q} \|q-p\|$. If $x_p\notin X_T$ and $y_p \notin Y_T$, then there are $x,x'\in X_T$ and $y,y' \in Y_T$ such that $x<x_p<x'$ and $y<y_p<y'$, as shown in Figure \ref{fig:TGmedians}(a).
%

Without loss of generality, we assume $|x_p-x|\geq \frac{1}{2}|x'-x|$ and $|y_p-y|\geq \frac{1}{2}|y-y'|$.
Since $p$ is the Tukey median of $Q$, we have $|Q\cap(-\infty, \infty)\times(-\infty,y_p]|\geq \frac{n}{3}$
where $(-\infty, \infty)\times(-\infty,y_p]=\{(x,y)\in \mathbb{R}^2|\ y\leq y_p\}$. Recall there are at most
$\lfloor \frac{n}{4}\rfloor$ points of $P_\all$ in $(-\infty,\infty)\times [y_p,y]$, which implies $|Q\cap(-\infty, \infty)\times(-\infty,y)|\geq \frac{n}{3}-\lfloor \frac{n}{4}\rfloor \geq \frac{n}{12}$. So, we have
\begin{equation*}
\cost(p,Q)\geq \frac{1}{n}\sum\nolimits_{q\in Q\cap(-\infty,\infty)\times (-\infty,y)}\|q-p\|\geq\frac{1}{n}\frac{n}{12}|y-y_p|\geq \frac{1}{24}|y-y'|.
\end{equation*}
Using a symmetric argument, we can obtain
$\cost(p,Q)\geq \frac{1}{24}|x-x'|$.

For any fixed $\varepsilon \in (0,1]$, and any two consecutive points $x,x'$ in $X_T$ we put
 $x_1,\cdots,x_{\lceil \frac{48}{\varepsilon}\rceil-1}$ into $X_T$ where $x_i=x+\frac{|x-x'|i}{\lceil \frac{48}{\varepsilon}\rceil}$. Also, for  any two consecutive point $y,y'$ in $Y_T$, we put
 $y_1,\cdots,y_{\lceil \frac{48}{\varepsilon}\rceil-1}$ into $Y_T$ where $y_i=y+\frac{|y-y'|i}{\lceil \frac{48}{\varepsilon}\rceil}$. So, for the Tukey median $p\in (x,x')\times (y,y')$, there exist $x_i\in X_T$
 and $y_j\in Y_T$ such that $|x_p-x_i|\leq \frac{\varepsilon}{48}|x-x'|$ and
 $|y_p-y_j|\leq \frac{\varepsilon}{48}|y-y'|$. Since we have shown that $\frac{1}{24}|x-x'|$ and $\frac{1}{24}|y-y'|$ are lower bounds for $\cost(p,Q)$, we obtain
\begin{equation*}
\begin{split}
\|(x_p,y_p)-(x_i,y_j)\|\leq &|x_p-x_i|+|y_p-y_j|\leq \frac{\varepsilon}{48}(|x-x'|+|y-y'|)\\
\leq & \frac{\varepsilon}{48}(24 \cost(p,Q)+24 \cost(p,Q))=\varepsilon \cost(p,Q).
\end{split}
\end{equation*}
Finally, we define $T$ as $T:=X_T\times Y_T$. Then for any $Q\Subset \Eu{P}$, if $p$ is the Tukey median of $Q$, there exists $t\in T$ such that $\|t-p\|\leq \varepsilon \cost(p,Q)$. Thus, $T$ is an $\varepsilon$-support for the Tukey median on $\Eu{P}$. Moreover, since $|X_T|=O(\frac{k}{\varepsilon})$ and $|Y_T|=O(\frac{k}{\varepsilon})$, we have $|T|=O(\frac{k^2}{\varepsilon^2})$.
\end{proof}

In a straight-forward extension, we can generalize the result of Theorem \ref{theorem the epsilon support for Tukey median} to $d$ dimensions.

\begin{theorem}
\label{theorem the epsilon support for Tukey median R^d}
Given a set of $n$ uncertain points $\Eu{P}=\{P_1,\cdots,P_n\}$, where $P_i=\{p_{i,1},\cdots,p_{i,k}\}$ $\subset\mathbb{R}^d$,
 and $\eps \in(0,1]$, we can construct an $\eps$-support $T$ for the Tukey median on $\Eu{P}$ that has a size $|T|=O((2d(d+1)(d+2)^2\frac{k}{\eps})^d)$.
\end{theorem}

\begin{figure}[b]
\vspace{-.2in}
  \includegraphics[width=0.3\textwidth]{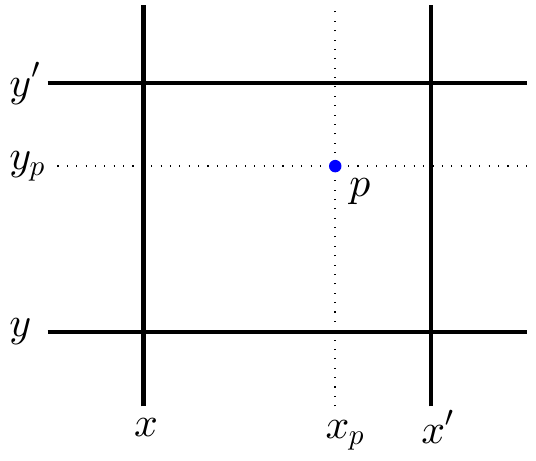}
  \;\;\;\;\;
  \includegraphics[width=0.3\textwidth]{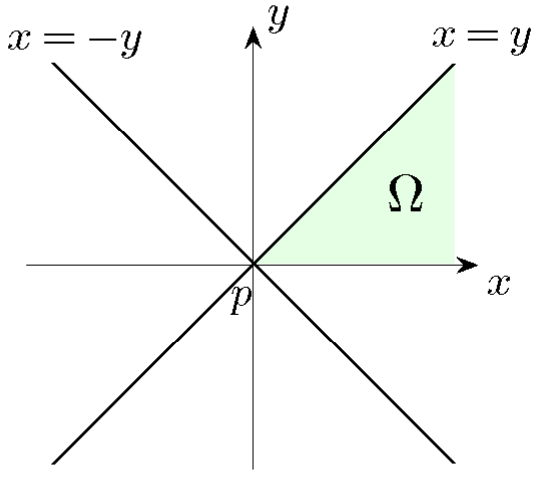}
  \;\;\;\;\;
  \includegraphics[width=0.3\textwidth]{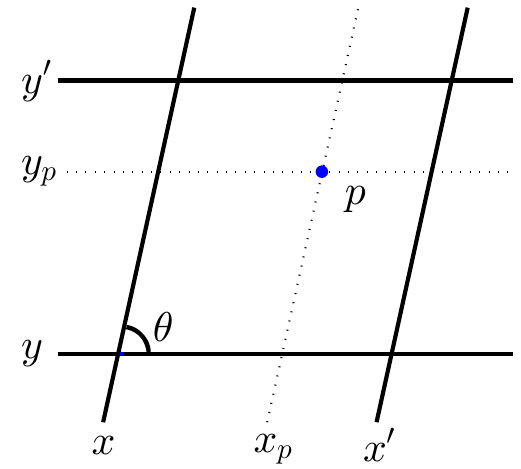}
\\
\phantom{ababababab} {\small (a) \hspace{1.6in} (b) \hspace{1.6in} (c) }

  \vspace{-0.1in}
  \caption{\label{fig:TGmedians}
    (a) Tukey median $p$ is in a grid cell formed by $x,x'$ and $y,y'$.
    (b) The plane is decomposed into 8 regions with the same shape.
    (c) Geometric median $p$ is in an oblique grid cell formed by $x,x'$ and $y,y'$.}
\end{figure}

\subsection{An $\eps$-Support for the Geometric Median}
\label{sec:geometric}
Unlike the Tukey median, 
there does not exist a constant $C>0$ such that:
for any geometric median $p$ of point set $Q\subset \mathbb{R}^d$,
any closed half space containing $p$ contains at least $\frac{1}{C}|Q|$ points of $Q$.
For example, suppose in $\mathbb{R}^2$ there are $2n+1$ points on $x$-axis with the median point at the origin; this point is also the geometric median.   If we move this point upward along the $y$ direction, then the geometric median also moves upwards.  However, for the line through the new geometric median and parallel to the $x$-axis, all $2n$ other points are under this line.

Hence, we need a new idea to adapt the method in Theorem \ref{theorem the epsilon support for Tukey median R^d} for the geometric median in $\mathbb{R}^d$.
We first consider the geometric median in $\mathbb{R}^2$. 
We show we can find \emph{some} line through it, such that on both sides of this line there are at least $\frac{n}{8}$ points.

\begin{lemma}
\label{lemma line through L1 median in R^2}
Suppose $p$ is the geometric median of $Q\subset \mathbb{R}^2$ with size $|Q|=n$.  There is a line $\ell$ through $p$ so
both closed half planes with $\ell$ as boundary contain at least $\frac{n}{8}$ points of $Q$.
\end{lemma}

\begin{proof}
We first build a rectangular coordinate system at the point $p$, which means $p$ is the origin with coordinates $(x_p,y_p)=(0,0)$.
Then we use the $x$-axis, $y$-axis and lines $x=y$, $x=-y$ to decompose the plane into eight regions, as shown in Figure~\ref{fig:TGmedians}(b). Since all these eight regions have the same shape, without loss of generality, we can assume $\Omega=\{(x,y)\in \mathbb{R}^2|\ x\geq y\geq 0\}$ contains the most points of $Q$. Then $|\Omega \cap Q|\geq \frac{n}{8}$, otherwise $n=|Q|=|\mathbb{R}^2 \cap Q|\leq 8 |\Omega \cap Q|<n$, which is a contradiction.

%

If $|Q\cap \{p\}|\geq \frac{n}{8}$, i.e., the multiset $Q$ contains $p$ at least $\frac{n}{8}$ times, then obviously this proposition is correct. So, we only need to consider the case
$|Q\cap \{p\}|< \frac{n}{8}$. We introduce notations $\widetilde{\Omega}=\Omega \setminus \{p\}$ and $\Omega^o= \Omega \setminus \partial \Omega$, and denote the coordinates of any $q \in Q$ as $q = (x_q, y_q)$.  From a property of the geometric median
(proven in Appendix \ref{A Property of Geometric Median})
we know
$\sum_{q\in Q\setminus\{p\}}\frac{x_q-x_p}{\|q-p\|} \leq |Q \cap \{p\}|$.  Since $|Q \cap \{p\}| <\frac{n}{8}$ this implies
\[
\sum\nolimits_{q\in Q\cap \widetilde{\Omega}}\frac{x_q}{\|q\|} +\sum\nolimits_{q\in Q \setminus \Omega}\frac{x_q}{\|q\|}
<\frac{n}{8},
\]
since $p$ is the origin and $Q\setminus\{p\}=(Q\cap \widetilde{\Omega}) \cup (Q\setminus \Omega) $.
From   $\frac{x_q}{\|q\|}=\frac{x_q}{\sqrt{x_q^2+y_q^2}}\geq \frac{1}{\sqrt{2}}$, $\forall q\in \widetilde{\Omega}$
we obtain
\[
|Q\cap \widetilde{\Omega}|\frac{1}{\sqrt{2}}\leq \sum\nolimits_{q\in Q\cap \widetilde{\Omega}}\frac{x_q}{\|q\|} <\frac{n}{8} - \sum\nolimits_{q\in Q \setminus \Omega}\frac{x_q}{\|q\|}\leq
  \frac{n}{8}+
| Q \setminus \Omega|\leq \frac{n}{8}+(n-| Q \cap \widetilde{\Omega}|)
\]
which implies there are not too many points in $\widetilde{\Omega}$,
\[
|Q\cap \widetilde{\Omega}|<\frac{\sqrt{2}n}{(1+\sqrt{2})} \cdot \frac{9}{8}<0.66n.
\]

Now, we define the two pairs of half spaces which share a boundary with $\widetilde{\Omega}$: $H_1^+=\{(x,y)\in \mathbb{R}^2|\ y\geq 0\}$, $H_1^-=\{(x,y)\in \mathbb{R}^2|\ y\leq 0\}$
and
$H_2^+=\{(x,y)\in \mathbb{R}^2|\ x-y\geq 0\}$, $H_2^-=\{(x,y)\in \mathbb{R}^2|\ x-y \leq 0\}$.
We assert either $|H_1^+\cap Q |\geq \frac{n}{8}$ and  $|H_1^- \cap Q |\geq \frac{n}{8}$, or $|H_2^+\cap Q |\geq \frac{n}{8}$ and  $|H_2^-\cap Q |\geq \frac{n}{8}$. Otherwise, since $|Q\cap \Omega|\geq \frac{n}{8}$ and $\Omega \subset H_1^+\cap H_2^+$,
we have  $|H_1^- \cap Q |<\frac{n}{8}$ and $|H_2^-\cap Q |< \frac{n}{8}$. From $H_1^-\cup H_2^- \cup \Omega^o=\mathbb{R}^2$ we have
\[
\begin{split}
n=&|Q|=|\mathbb{R}^2\cap Q|=|(H_1^-\cup H_2^- \cup \Omega^o)\cap Q|\leq
|H_1^-\cap Q|+| H_2^- \cap Q|+|\Omega^o\cap Q|\\
\leq&|H_1^-\cap Q|+| H_2^- \cap Q|+|\widetilde{\Omega}\cap Q|
\leq \frac{n}{8}+\frac{n}{8}+0.66n<n,
\end{split}
\]
which is a contradiction.
Therefore, among lines $\ell_1:  y=0$ and $\ell_2:  x-y=0$, which both go through $p$, one of them has at least $n/8$ points from $Q$ on both sides.
\end{proof}

\begin{theorem}
\label{theorem the epsilon support for L_1 median in R^2}
Given a set of $n$ uncertain points $\Eu{P}=\{P_1,\cdots,P_n\}$, where $P_i=\{p_{i,1},\cdots,p_{i,k}\}$ $\subset\mathbb{R}^2$,
 and $\eps \in(0,1]$, we can construct an $\eps$-support $T$ for the geometric median on $\Eu{P}$ that has a size $|T|=O(\frac{k^2}{\eps^2})$.
\end{theorem}

\begin{proof}
The idea to prove this theorem is to use several oblique coordinate systems.
We consider an oblique coordinate system, the angle between $x$-axis and $y$-axis is $\theta \in (0,\frac{\pi}{2}]$, and use the technique in Theorem  \ref{theorem the epsilon support for Tukey median} to generate a grid.
More precisely, we project $P_\all$ onto the $x$-axis
along the $y$-axis of the oblique coordinate system
to obtain a set $X$, sort all points in $X$, and choose
 one point from $X$ every $\lfloor \frac{n}{9}\rfloor$ points to form a set $X_T$.   Then we use the same method to generate $Y$ and $Y_T$
projecting along the $x$-axis in the oblique coordinate system.
For each point $x\in X_T$ we draw a line through $(x,0)$ parallel to the (oblique) $y$-axis, and
for each point $y\in Y_T$ we draw a line through $(0,y)$ parallel to the (oblique) $x$-axis.
%

Let $p$ with coordinates $(x_p,y_p)$ be the geometric median of some traversal $Q\Subset \Eu{P}$ and
$\cost(p,Q)=\frac{1}{n}\sum_{q\in Q} \|q-p\|$. If $x_p\notin X_T$ and $y_p \notin Y_T$, then there are $x,x'\in X_T$ and $y,y' \in Y_T$ such that $x_p\in(x,x')$ and $y_p\in(y,y')$, as shown in Figure \ref{fig:TGmedians}(c).

If we have the condition:
\begin{equation}  \label{conditions for L1 median}
\begin{split}
|Q\cap(-\infty, \infty)\times(-\infty,y_p]|\geq \frac{n}{8},\ \ |Q\cap(-\infty, \infty)\times[y_p,\infty)|\geq \frac{n}{8},\\
|Q\cap(-\infty,x_p]\times (-\infty, \infty)|\geq \frac{n}{8},\ \ |Q\cap[x_p,\infty)\times (-\infty, \infty)|\geq \frac{n}{8},\\
\end{split}
\end{equation}
then we can make the following computation.

Without loss of generality, we assume $|x_p-x|\geq \frac{1}{2}|x'-x|$ and $|y_p-y|\geq \frac{1}{2}|y-y'|$.
There are at most
$\lfloor \frac{n}{9}\rfloor$ points of $P_\all$ in $(-\infty,\infty)\times [y_p,y]$, which implies $|Q\cap(-\infty, \infty)\times(-\infty,y)|\geq \frac{n}{8}-\lfloor \frac{n}{9}\rfloor \geq \frac{n}{72}$.
 So, we have
\begin{equation*}
\cost(p,Q)\geq \frac{1}{n}\sum\nolimits_{q\in Q\cap(-\infty,\infty)\times (-\infty,y)}\|q-p\|\geq\frac{1}{n}\frac{n}{72}|y-y_p|\geq \frac{ \sin(\theta)}{144}|y-y'|.
\end{equation*}
 Similarly, we can prove
$\cost(p,Q)\geq \frac{ \sin(\theta)}{144}|x-x'|$.

For any fixed $\varepsilon \in (0,1]$, and any two consecutive points $x,x'$ in $X_T$ we put
 $x_1,\cdots,x_{\lceil \frac{288}{\varepsilon \sin(\theta)}\rceil-1}$ into $X_T$ where $x_i=x+\frac{|x-x'|i}{\lceil \frac{288}{\varepsilon \sin(\theta)}\rceil}$. Also, for  any two consecutive point $y,y'$ in $Y_T$, we put
 $y_1,\cdots,y_{\lceil \frac{288}{\varepsilon \sin(\theta)}\rceil-1}$ into $Y_T$ where $y_i=y+\frac{|y-y'|i}{\lceil \frac{288}{\varepsilon\sin(\theta)}\rceil}$. So, for the $L_1$ median $p\in (x,x')\times (y,y')$, there exist $x_i\in X_T$
 and $y_j\in Y_T$ such that $|x_p-x_i|\leq \frac{\varepsilon \sin(\theta)}{288}|x-x'|$ and
 $|y_p-y_j|\leq \frac{\varepsilon \sin(\theta)}{288}|y-y'|$. Since we have shown that both $\frac{\sin (\theta)}{144}|x-x'|$ and $\frac{\sin(\theta)}{144}|y-y'|$ are lower bounds for $\cost(p,Q)$, using the distance formula in an oblique coordinate system, we have
\begin{equation*}
\begin{split}
\|(x_p,y_p)-(x_i,y_j)\|\leq& ((x_p-x_i)^2+(y_p-y_j)^2+2(x_p-x_i)(y_i-y_p)\cos(\theta))^{\frac{1}{2}}\\
\leq &((x_p-x_i)^2+(y_p-y_j)^2+2|x_p-x_i||y_i-y_p|)^{\frac{1}{2}}=|x_p-x_i|+|y_p-y_j|\\
\leq &\frac{\varepsilon\sin(\theta)}{288}(|x-x'|+|y-y'|)\\
\leq & \frac{\varepsilon\sin(\theta)}{288}\left(\frac{144}{\sin(\theta)} \cost(p,Q)+\frac{144}{\sin(\theta)} \cost(p,Q)\right)=\varepsilon \cost(p,Q).
\end{split}
\end{equation*}
Therefore, if all $k^n$ geometric medians of traversals satisfy \eqref{conditions for L1 median} and $\theta \in(0,\frac{\pi}{2}]$ is a constant then
$T=X_T\times Y_T$ is an $\varepsilon$-support of size $O\left(\frac{k^2}{(\sin(\theta)\varepsilon)^2}\right)$ for the geometric median on $\Eu{P}$.

Although we cannot find an oblique coordinate system to make  \eqref{conditions for L1 median} hold for all $k^n$ medians,
we can use several oblique coordinate systems. Using the result of Lemma \ref{lemma line through L1 median in R^2},  for any geometric median of $n$ points $Q$,
we know there exists a line $\ell$ through $p$ and parallel to a line in $\{\ell_1:y=0,\ \ell_2:x-y=0,\ \ell_3: x=0,\ \ell_4:x+y=0\}$, such that in both sides of this line, there are at least $\frac{n}{8}$ points of $Q$. Since we did not make any assumption on the distribution of points in $Q$, if we rotate $\ell_1,\ell_2,\ell_3,\ell_4$  anticlockwise by $\frac{\pi}{8}$
 around the origin, we can obtain four lines  $\ell_1',\ell_2',\ell_3',\ell_4'$, and there exists a line $\ell'$ through $p$ and parallel to a line in $\{\ell_1',\ell_2',\ell_3',\ell_4'\}$, such that on both sides of this line, there are at least $\frac{n}{8}$ points of $Q$. The angle between $\ell$ and $\ell'$ is at least $\frac{\pi}{8}$.

Therefore, given $\Eu{L}=\{\ell_1,\ell_2,\ell_3,\ell_4\}$  and $\Eu{L}'=\{\ell_1',\ell_2',\ell_3',\ell_4'\}$,
for each pair $(\ell,\ell')\in \Eu{L}\times\Eu{L}'$, we take $\ell$ and $\ell'$ as $x$-axis and $y$-axis respectively to build an oblique coordinate system, and then use the above method to compute a set $T(\ell,\ell')$. Since for any geometric median $p$
there must be an oblique coordinate system based on some $(\ell,\ell')\in \Eu{L}\times\Eu{L}'$ to make \eqref{conditions for L1 median} hold for $p$, we can take $T=\cup_{\ell\in \Eu{L},\ell'\in\Eu{L}'}T(\ell,\ell')$ as an $\eps$-support for geometric median on $\Eu{P}$, and the size of $T$ is $|T|=O\left(16\frac{k^2}{(\sin(\frac{\pi}{8})\eps)^2}\right)=O\left(\frac{k^2}{\eps^2}\right)$.
\end{proof}

The result of Theorem \ref{theorem the epsilon support for L_1 median in R^2} can be generalized to $\mathbb{R}^d$ and details are in
Appendix \ref{Size bound of T in R^d}.

\subsection{Assigning a Weight to $T$ in $\mathbb{R}^1$}
\label{sec:weight-T}

Here we provide an algorithm to assign a weight to $T$ in $\bb{R}^1$,
which approximates the probability distribution of median.
For $T$ in $\mathbb{R}^d$, we provide a randomized algorithm in Section \ref{sec:rand}.

Define the weight of $p_{i,j}\in P_\all$ as $w(p_{i,j})=\frac{1}{k^n}|\{Q\Subset \Eu{P} \mid p_{i,j} \text{ is the median of Q}\}|$, the probability it is the median. Suppose $T$ is constructed by our greedy algorithm for $\bb{R}^1$.
For $p_{i,j}\in P_\all$,  we introduce a map $f_T:P_\all\rightarrow T$,
\[
f_T(p_{i,j})= \text{arg}\min\{|x-p_{i,j}| \mid x\in T, \ |x-p_{i,j}|\leq \eps \cost(p_{i,j}) \},
\]
where $\cost(p_{i,j})=\min \{\cost(p_{i,j},Q)\ |\ p_{i,j} \text{ is the median of $Q$ and } Q \Subset \Eu{P} \}$.

Intuitively, this maps each $p_{i,j} \in P_\all$ onto the closest point $x \in T$, unless it violates the $\eps$-approximation property which another further point satisfies.

Now for each $x\in T$, define weight of $x$ as
$\hat{w}(x)=\sum_{\{p_{i,j}\in P_\all \mid f_T(p_{i,j})=x\}}w(p_{i,j})$.
So we first compute the weight of each point in $P_\all$ and then obtain the weight of points in $T$ in another linear sweep.   Our ability to calculate the weights $w$ for each point in $P_\all$ is summarized in the next lemma.
The algorithm, explained within the proof, is a dynamic program that expands a specific polynomial similar to Li \etal~\cite{LSD09}, where in the final state, the coefficients correspond with the probability of each point being the median.

\begin{lemma}    \label{theorem the weight of points in P_all}
We can output $w(p_{i,j})$ for all points in $P_\all$ in $\bb{R}^1$ in $O(n^2k)$ time.
\end{lemma}

\begin{proof}
For any $p_{i_0}\in P_{i_0}$, we define  the following terms to count the number of points to the left ($l_j$) or right ($r_j$) of it in the $j$th uncertain point (excluding $P_{i_0}$):
\[
l_j=\begin{cases}
       |\{p\in P_j \mid p\leq p_{i_0}\}|  & \text{if }  1\leq j\leq i_0-1 \\
       |\{p\in P_{j+1} \mid p\leq p_{i_0}\}|  & \text{if }  i_0\leq j\leq n-1 \\
\end{cases},
\;\;
r_j=\begin{cases}
       |\{p\in P_j \mid p\geq p_{i_0}\}|  & \text{if }  1\leq j\leq i_0-1 \\
       |\{p\in P_{j+1} \mid p\geq p_{i_0}\}|  & \text{if }  i_0\leq j\leq n-1 \\
\end{cases}.
\]
Then,
if $n$ is odd, we can write the weight of $p_{i_0}$ as
\begin{equation*}
w(p_{i_0})=\frac{1}{k^n}\sum_{\substack{S_1\cap S_2=\emptyset\\ S_1\cup S_2=\{1,\cdots,n-1\} }}
(l_{i_1} \cdot l_{i_2}\cdot \ldots \cdot l_{i_{\frac{n-1}{2}}} \cdot r_{j_1}\cdot r_{j_2} \cdot \ldots \cdot r_{j_{\frac{n-1}{2}}}),
\end{equation*}
where $S_1=\{i_1,i_2,\cdots,i_{\frac{n-1}{2}}\}$ and $S_2=\{j_1,j_2,\cdots,j_{\frac{n-1}{2}}\}$.
This sums over all partitions $S_1, S_2$ of uncertain points on the left or right of $p_{i_0}$ for which it is the median, and each term is the product of ways each uncertain point can be on the appropriate side.
We define $w(p_{i_0})$ similarly when $n$ is even, then the last index of $S_2$ is $j_{\frac{n}{2}}$.

We next describe the algorithm for $n$ odd; the case for $n$ even is similar. To compute $\sum_{\substack{S_1\cap S_2=\emptyset\\ S_1\cup S_2=\{1,\cdots,n-1\} }}
(l_{i_1} \cdot l_{i_2}\cdot \ldots \cdot l_{i_{\frac{n-1}{2}}} \cdot r_{j_1}\cdot r_{j_2} \cdot \ldots \cdot r_{j_{\frac{n-1}{2}}})$, we consider the following polynomial:
\begin{equation}\label{polynomial}
(l_1x+r_1)(l_2x+r_2)\cdots (l_{n-1}x+r_{n-1}),
\end{equation}
where $\sum_{\substack{S_1\cap S_2=\emptyset\\ S_1\cup S_2=\{1,\cdots,n-1\} }}
(l_{i_1} \cdot l_{i_2}\cdot \ldots \cdot l_{i_{\frac{n-1}{2}}} \cdot r_{j_1}\cdot r_{j_2} \cdot \ldots \cdot r_{j_{\frac{n-1}{2}}})$ is the coefficient of $x^{\frac{n-1}{2}}$. We define $\rho_{i,j}$ $(1\leq i\leq n-1, 0\leq j \leq i)$ as the coefficient of $x^j$ in the polynomial
$(l_1x+r_1)\cdots (l_{i}x+r_{i})$ and then it is easy to check $\rho_{i,j}=l_i\rho_{i-1,j-1}+r_i\rho_{i-1,j}$.   Thus we can use dynamic programming to compute $\rho_{n-1,0},\rho_{n-1,1},\cdots,\rho_{n-1,n-1}$, as shown in Algorithm~\ref{algorithm for coefficient}.

\begin{algorithm}
\caption{Compute  $\rho_{n-1,0},\rho_{n-1,1},\cdots,\rho_{n-1,n-1}$}
\label{algorithm for coefficient}
\begin{algorithmic}
\STATE {Let $\rho_{1,0}=r_1,\rho_{1,1}=l_1,\rho_{1,2}=0$.}
\FOR {$i=2$ to $n-1$}
\FOR {$j=0$ to $i$}
\STATE $\rho_{i,j}=l_i\rho_{i-1,j-1}+r_i\rho_{i-1,j}$
\ENDFOR
\STATE $\rho_{i,i+1}=0$
\ENDFOR
\RETURN $\rho_{n-1,0},\rho_{n-1,1},\cdots,\rho_{n-1,n-1}$.
\end{algorithmic}
\end{algorithm}

Thus Algorithm~\ref{algorithm for coefficient} computes the weight $\frac{1}{k^n} w(p_{i_0}) = \rho_{n-1,\frac{n-1}{2}}$ for a single $p_{i_0} \in P_\all$.  Next we show, we can reuse much of the structure to compute the weight for another point; this will ultimately shave a factor $n$ off of running Algorithm~\ref{algorithm for coefficient} $nk$ times.

Suppose for $p_{i_0}\in P_{i_0}$ we have obtained $\rho_{n-1,0},\rho_{n-1,1},\ldots,\rho_{n-1,n-1}$ by Algorithm~\ref{algorithm for coefficient}, and then we consider $p_{i_0'}=\min\{p\in P_\all \setminus P_{i_0} \mid p\geq p_{i_0}\}$. We assume $p_{i_0'}\in P_{i_0'}$, and if $i_0'<i_0$, we construct a polynomial
\begin{equation}  \label{polynomial i'<i}
(l_1x+r_1)\cdots (l_{i_0'-1}x+r_{i_0'-1})(\tilde{l}_{i_0'}x+\tilde{r}_{i_0'})(l_{i_0'+1}x+r_{i_0'+1})\cdots (l_{n-1}x+r_{n-1})
\end{equation}
and if $i_0'>i_0$, we construct a polynomial
\begin{equation}  \label{polynomial i'>i}
(l_1x+r_1)\cdots (l_{i_0'-2}x+r_{i_0'-2})(\tilde{l}_{i_0'-1}x+\tilde{r}_{i_0'-1})(l_{i_0'}x+r_{i_0'})\cdots (l_{n-1}x+r_{n-1})
\end{equation}
where $\tilde{l}_{i_0'}=\tilde{l}_{i_0'-1}=|\{p\in P_{i_0} \mid p\leq p_{i_0'}\}|$ and $\tilde{r}_{i_0'}=\tilde{r}_{i_0'-1}=|\{p\in P_{i_0} \mid p\geq p_{i_0'}\}|$.

Since \eqref{polynomial} and \eqref{polynomial i'<i} have only one different factor, we
obtain the coefficients of  \eqref{polynomial i'<i} from the coefficients of \eqref{polynomial} in $O(n)$ time.  We recover the coefficients of $(l_1x+r_1)\cdots (l_{i'-1}x+r_{i'-1})(l_{i_0'+1}x+r_{i_0'+1})\cdots (l_{n-1}x+r_{n-1})$ from  $\rho_{n-1,0},\rho_{n-1,1},\cdots,\rho_{n-1,n-1}$, and then use these coefficients to compute the coefficients of \eqref{polynomial i'<i}.
Similarly,
if $i_0'>i_0$, we
obtain the coefficients of  \eqref{polynomial i'>i} from the coefficients of \eqref{polynomial}.
Therefore, we can use $O(n^2)$ time to compute the weight of the first point in $P_\all$ and then use $O(n)$ time to compute the weight of each other point. The whole time is $O(n^2)+nkO(n)=O(n^2k)$.
 \end{proof}

\begin{corollary}
We can assign $\hat w(x)$ to each $x \in T$ in $\bb{R}^1$ in $O(n^2 k)$ time.
\end{corollary}

\section{A Randomized Algorithm to Construct a Covering Set}
\label{sec:rand-alg}
In this section we describe a much more general randomized algorithm for robust estimators on uncertain data.
It constructs an approximate covering set of the support of the distribution of the estimator, and estimates the weight at the same time.  The support of the distribution is not as precise compared to the techniques in the previous section in that the new technique may fail to cover regions with small probability of containing the estimator.

Suppose $\Eu{P}=\{P_1,P_2,\cdots,P_n\}$ is a set of uncertain data, where for $i\in[n]$, $P_i=\{p_{i,1},p_{i,2},\cdots$, $p_{i,k}\}$ $\subseteq \Eu{X}$ for some domain $\Eu{X}$. An estimator $E:\ \{Q\ |\ Q\Subset \Eu{P}\} \mapsto Y$
maps $Q \Subset \Eu{P}$ to a metric space $(Y,\varphi)$.
Let $B(y,r)=\{y'\in Y|\ \varphi(y,y')\leq r\}$ be a ball of radius $r$ in that metric space.  We denote $\nu$ as the VC-dimension of the range space $(Y,\Eu{R})$ induced by these balls, with $\Eu{R}=\{B(y,r) \mid y\in Y, r\geq 0\}$.

We now analyze the simple algorithm which randomly instantiates traversals $Q \Subset \Eu{P}$, and constructors their estimators $z = E(Q)$.  Repeating this $N$ times builds a domain $T = \{z_1, z_2, \ldots, z_N\}$ each with weight $w(z_i) = 1/N$.  Duplicates of domain points can have their weights merged as described in Algorithm \ref{alg:ran-weights1}.
%
%
%
%
\begin{algorithm}[H]
\caption{Approximate the weight of points in $T$}  \label{alg:ran-weights1}
\begin{algorithmic}
\STATE Initialize $T = \emptyset$
\FOR {$j=1$ to $N$}
\STATE Randomly choose $Q \Subset \Eu{P}$, and set $z = E(Q)$.
\STATE \textbf{if} $z  = z'$ for some $z' \in T$, \textbf{then} increment $c_{z'} = c_{z'}+1$
\STATE \textbf{else}  add $z$ to $T$, and set $c_z = 1$.
\ENDFOR
\RETURN $\frac{c_z}{N}$ as the approximate value of $w(z)$ for all $z \in T$
\end{algorithmic}
\end{algorithm}

\begin{theorem}  \label{thm:weight}
For $\eps >0$ and $\delta \in (0,1)$, set $N = O((1/\eps^2) (\nu + \log(1/\delta)))$.  
Then, with probability at least $1-\delta$, for any $B \in \Eu{R}$ we have
$
\left| \sum_{z \in T \cap B} w(z)  - \Pr_{Q \Subset \Eu{P}}[E(Q) \in B] \right| \leq \eps.
$
\end{theorem}

\begin{proof}
Let $T^*$ be the true support of $E(Q)$ where $Q \Subset P$, and let $w^* : T^* \to \bb{R}^+$ be the true probability distribution defined on $T^*$; e.g., for discrete $T^*$, then for any $z' \in T^*$,  $w^*(z') = \Pr_{Q \Subset P}[E(Q) = z']$.  Then each random $z$ generated is a random draw from $w^*$.  Hence for a range space with bounded VC-dimension~\cite{VC71} $\nu$, we can apply the sampling bound~\cite{LLS01} for $\eps$-approximations of these range spaces to prove our claim.
\end{proof}

In Theorem \ref{thm:weight}, for $z_i\in T$, if we choose $B = B(z_i, r) \in \Eu{R}$ with $r$ small enough such that $T\cap B$ only contains $z_i$, then we obtain the following.

\begin{corollary} \label{cor:weight}
For $\eps >0$ and $\delta \in (0,1)$, set $N = O((1/\eps^2) (\nu + \log(1/\delta)))$.
Then, with probability at least $1-\delta$, for any $z \in Y$ we have
$
\left| w(z)  - \Pr_{Q \Subset \Eu{P}}[E(Q) = z] \right| \leq \eps.
$
\end{corollary}

\begin{myremark}
We can typically define a metric space $(Y, \varphi)$ where $\nu = O(1)$; for instance for point estimators (e.g., the geometric median), define a projection into $\bb{R}^1$ so no $z_i$s map to the same point, then define distance $\varphi$ as restricted to the distance along this line, so metric balls are intervals (or slabs in $\bb{R}^d$);  these have $\nu = 2$.
\end{myremark}


\subsection{Application to Geometric Median}
\label{sec:rand}
For each $Q \Subset \Eu{P}$, the geometric median $m_Q$ may take a distinct value.  Thus even calculating that set, let alone their weights in the case of duplicates, would require at least $\Omega(k^n)$ time.
But it is straightforward to apply this randomized approach.
For $P_\all \in \bb{R}^d$, the natural metric space $(Y, \varphi)$ is $Y = \bb{R}^d$ and $\varphi$ as the Euclidian distance.

However, there is no known closed form solution for the geometric median; it can be computed within any additive error $\phi$ through various methods~\cite{Wei37,CT90,BMM03,ARR98}.  As such, we can state a slightly more intricate corollary.

\begin{corollary}
Set $\eps >0$ and $\delta \in (0,1)$ and $N = O((1/\eps^2) (d + \log(1/\delta)))$.  For an uncertain point set $\Eu{P}$ with $P_\all \subset \bb{R}^d$, let the estimator $E$ be the geometric median, and let $E_\phi$ be an algorithm that finds an approximation to the geometric median within additive error $\phi > 0$.  Run the algorithm using $E_\phi$.  Then for any ball $B = B(x,r) \in \Eu{R}$, there exists%
\footnote{
To simplify the discussion on degenerate behavior, define ball $B'$, so any point $q$ on its boundary can be defined inside or outside of $B$, and this decision can be different for each $q$, even if they are co-located.
}
 another ball $B' = B(x,r')$ with $|r-r'| \leq \phi$ such that with probability at least $1-\delta$,
\[
\Big| \sum_{z \in T \cap B'} w(z)  - \Pr_{Q \Subset \Eu{P}}[E(Q) \in B] \Big| \leq \eps.
\]
\end{corollary}

\subsection{Application to Siegel Estimator}
The Siegel (repeated median) estimator~\cite{Sie82} is a robust estimator $S$ for linear regression in $\bb{R}^2$ with optimal breakdown point $0.5$.  For a set of points $Q$, for each $q_i \in Q$ it computes slopes of all lines through $q_i$ and each other $q' \in Q$, and takes their median $a_i$.  Then it takes the median $a$ of the set $\{a_i\}_i$ of all median slopes.  The offset $b$ of the estimated line $\ell : y = a x + b$, is the median of $(y_i - a x_i)$ for all points $q_i = (x_i, y_i)$.
For uncertain data $P_\all \subset \bb{R}^2$, we can directly apply our general technique for this estimator.

We use the following metric space $(Y,\varphi)$.  Let
$Y = \{\ell \mid \ell \text{ is a line in } \mathbb{R}^2 \text{ with form } y=ax+b, \text{ where } a,b \in \mathbb{R}\}$.  Then let $\varphi$ be the Euclidean distance in the standard dual; for two lines $\ell : y = ax+b$ and $\ell' : y = a' x + b'$, define
$
\varphi(\ell, \ell') = \sqrt{(a-a')^2 + (b-b')^2}.
$
By examining the dual space, we see that $(Y,\Eu{R})$ with $\Eu{R} = \{B(\ell, r) \mid \ell \in Y, r \geq 0\}$ and $B(\ell,r) = \{\ell' \in Y \mid \varphi(\ell,\ell') \leq r\}$ has a VC-dimension $3$.

From the definition of the Siegel estimator~\cite{Sie82}, there can be at most $O(n^3 k^3)$ distinct lines in $T = \{S(Q) \mid Q \Subset \Eu{P}\}$.  By Corollary \ref{cor:weight}, setting $N = O((1/\eps^2) \log(1/\delta))$, then with probability at least $1-\delta$ for all $z \in T$ we have
$
\Big| w(z) - \Pr_{Q \Subset \Eu{P}} [S(Q) = z] \Big| \leq \eps.
$


\section{Conclusion}
We initiate the study of robust estimators for uncertain data, by studying the median, as well as extensions to the geometric median and Siegel estimators, on locationally uncertain data points.  We show how to efficiently create approximate distributions for the location of these medians in $\bb{R}^1$.  We generalize these approaches to robust estimators associated with bounded VC-dimension range spaces in a general metric space.  We also argue that although we can use such distributions to calculate a single-point representation of these distributions, it is not very stable to the input distributions, and serves as a poor representation when the true scenario is multi-modal; hence further motivating our distributional approach.

\subparagraph*{Acknowledgements. }
The authors would like to thank anonymous reviewers for helping simplify some proofs, and for the Shonan Village Center where some of this work took place.

\newpage
\appendix%

\section{The Size of $\hat{T}$ Based on $\hat{\cost}$}  \label{the size of hat{T}}

For a given positive number $\varepsilon$ and a set of uncertain points $\Eu{P} = \{P_1, \cdots, P_n\}$ where $P_i=\{p_{i,1},\cdots,p_{i,k}\}\subset \mathbb{R}$, $i\in[n]$, if we define
$\hat \cost(x) = \frac{1}{n} \sum_{i=1}^n \min_{1 \leq j \leq k} |x - p_{i,j}|$ and try to find a set $\hat{T}$ such that for any $Q\Subset \Eu{P}$ , there exists $x\in \hat{T}$ s.t. $|x-m_Q|\leq \varepsilon \hat{\cost}(m_Q)$,  then for some fixed $\varepsilon>0$, the size of $\hat{T}$ may satisfy $|\hat{T}|=\Omega(nk)$.

In fact, for this data set: $\eps=\frac{1}{4}$, $k=2$, $p_{i,1}=1-\frac{1}{2^{i-1}}$ and $p_{i,2}=1$ for all $i\in[n]$,
 we have
\begin{equation*}
\begin{split}
\hat{\cost}(p_{i,1})=&\frac{1}{n}\left(\sum_{j=1}^{i-1}(p_{j,2}-p_{i,1})+\sum_{j=i+1}^n(p_{j,1}-p_{i,1})\right)\\
=&\frac{1}{n}\left(\sum_{j=1}^{i-1}\big(1-(1-\frac{1}{2^{i-1}})\big)+\sum_{j=i+1}^n\big(1-\frac{1}{2^{j-1}}-(1-\frac{1}{2^{i-1}})\big)\right)\\
=&\frac{1}{2^{i-1}}+\frac{1}{n}\big(\frac{1}{2^{n-1}}-2\frac{1}{2^{i-1}}\big)<\frac{1}{2^{i-1}},
\end{split}
\end{equation*}
which implies
\begin{equation*}
\varepsilon\hat{\cost}(p_{i,1})+\varepsilon\hat{\cost}(p_{i+1,1})<\frac{1}{4}\frac{1}{2^{i-1}}+\frac{1}{4}\frac{1}{2^i}
<\frac{1}{2^i}=p_{i+1,1}-p_{i,1}.
\end{equation*}
So we have $[p_{i,1}-\varepsilon\hat{\cost}(p_{i,1}),p_{i,1}+\varepsilon\hat{\cost}(p_{i,1})]
\cap[p_{i+1,1}-\varepsilon\hat{\cost}(p_{i+1,1}),p_{i+1,1}+\varepsilon\hat{\cost}(p_{i+1,1})]=\emptyset$ for $i\in[n]$, which implies $|\hat{T}|\geq n$.

Now, if we consider $n=1,2,3,\cdots$, $k=2,4,6,\cdots$ and $p_{i,j}=\frac{1}{2}(3j-1)-\frac{1}{2^{i-1}}$, $p_{i,j+1}=\frac{1}{2}(3j-1)$
for $j=1,3,5, \cdots k-1$ and $i\in[n]$, then is easy to check $|\hat{T}|\geq \frac{1}{2}kn$. Therefore, we have $|\hat{T}|=\Omega(nk)$.

\section{A Property of Geometric Median}
\label{A Property of Geometric Median}

To prove the result of Lemma \ref{lemma line through L1 median in R^2}, we need the following property of geometric median.  Although this result is stated on Wikipedia, we have not found a proof in the literature, so we present it here for completeness.

\begin{lemma} \label{A Property of Geometric Median}
Suppose $p$ is the geometric median of $Q=\{q_1,\cdots,q_n\}\subset \mathbb{R}^d$, and $(x_1,\cdots,x_d)$ and $(x_{i,1},\cdots,x_{i,d})$ are the coordinates of $p$ and $q_i$ respectively, then we have
$|\sum_{q_i\in Q\setminus\{p\}}\frac{x_j-x_{i,j}}{\|q-p\|}| \leq |Q \cap \{p\}|$ for any $j\in[d]$.
\end{lemma}

\begin{proof}

We introduce the notation $f(y)=f_1(y)+f_2(y)$ where $f_1(y)=\sum_{q\in Q\setminus\{p\}}\|q-y\|$ and $f_2(y)=\sum_{q\in Q\cap\{p\}}\|q-y\|$. Suppose $v_j\in \mathbb{R}^d$ is a vector such that its $j$-th component is one and all other components are zero.
 Since $p$ is the global minimum point of $f$, for any $j\in [d]$ there exists $\delta_j>0$ such that
\begin{equation*}
f(p+\varepsilon v_j)\geq f(p) \text{ and } f(p-\varepsilon v_j)\geq f(p),\ \ \ \forall \ \varepsilon\in[0,\delta_j),
\end{equation*}
which implies
\begin{equation} \label{f1+f2>=f2 1}
f_1(p+\varepsilon v_j)+f_2(p+\varepsilon v_j)\geq f_1(p)+f_2(p), \ \  \forall \ \varepsilon\in[0,\delta_j),
\end{equation}
and
\begin{equation} \label{f1+f2>=f2 2}
f_1(p-\varepsilon v_j)+f_2(p-\varepsilon v_j)\geq f_1(p)+f_2(p), \ \  \forall \ \varepsilon\in[0,\delta_j).
\end{equation}

Since $f_2(p)=0$, from \eqref{f1+f2>=f2 1} we have
$\frac{1}{\varepsilon} (f_1(p+\varepsilon v_j)-f_1(p))\geq -\frac{1}{\varepsilon}f_2(p+\varepsilon v_j)=-|Q\cap\{p\}|$.
Letting $\varepsilon\rightarrow 0+$, we obtain $\frac{\partial f_1(p)}{\partial x_j}\geq -|Q\cap\{p\}|$ which implies
\begin{equation}  \label{lower bound of sum frac{x_j-x_{i,j}}{|q-p|}}
\sum_{q_i\in Q\setminus\{p\}}\frac{x_j-x_{i,j}}{\|q-p\|}\geq -|Q\cap\{p\}|.
\end{equation}
Similarly, using \eqref{f1+f2>=f2 2} we can obtain
\begin{equation}  \label{upper bound of sum frac{x_j-x_{i,j}}{|q-p|}}
\sum_{q_i\in Q\setminus\{p\}}\frac{x_j-x_{i,j}}{\|q-p\|}\leq |Q\cap\{p\}|.
\end{equation}

Thus, conclusion of this lemma is obtained from \eqref{lower bound of sum frac{x_j-x_{i,j}}{|q-p|}}
and \eqref{upper bound of sum frac{x_j-x_{i,j}}{|q-p|}}.
\end{proof}

The bound in Lemma \ref{A Property of Geometric Median} is tight. For example, we consider $Q=\{(-2,0), (-1,0),$ $(0,0), (1,0),(-1,1),(-1,-1)\}\subset \mathbb{R}^2$, then $p=(-1,0)$ is the geometric median of $Q$ and $|\sum_{q=(x_q,y_q)\in Q\setminus\{p\}}\frac{-1-x_q}{\|q-p\|}|=1=|Q\cap{p}|$.

\section{Size bound of $T$ in $\mathbb{R}^d$} \label{Size bound of T in R^d}

Using the method in the proof of Theorem \ref{theorem the epsilon support for L_1 median in R^2}, we can generalize
 the result of this theorem to $\mathbb{R}^3$ and higher dimensional space.

\begin{theorem}
\label{theorem the epsilon support for L_1 median in R^3}
Given a set of $n$ uncertain points $\Eu{P}=\{P_1,\cdots,P_n\}$, where $P_i=\{p_{i,1},\cdots,p_{i,k}\}$ $\subset\mathbb{R}^3$,
 and $\eps \in(0,1]$, we can construct an $\eps$-support $T$ for $L_1$ median on $\Eu{P}$ that has a size $|T|=O\left(\frac{k^3}{\eps^3}\right)$.
\end{theorem}

\begin{proof}

The first step is to obtain a result similar to Lemma \ref{lemma line through L1 median in R^2}:
if $p$ is the $L_1$ median of a set of $n$ points $Q\subset \mathbb{R}^3$, then we can find a plane $h$ through $p$, such that any closed half space with $h$ as its boundary contains at least $\frac{n}{24}$ points of $Q$.

To prove this, we build a rectangular coordinate system at the point $p$, and use nine planes
$\Eu{H}_3=\{x_1=0, x_2=0, x_3=0, x_1\pm x_2=0, x_2\pm x_3=0, x_3\pm x_1=0\}$ to partition
$\mathbb{R}^3$ into 24 regions: $\{\Omega_{i,\mathbf{s}}|\ i\in\{1,2,3\}, \mathbf{s}\in\{1,-1\}^3\}$,
where $\Omega_{i,\mathbf{s}}=\Omega_{i,(s_1,s_2,s_3)}:=\{(x_1,x_2,x_3)\in \mathbb{R}^3|\ s_ix_i\geq s_jx_j\geq 0,\text{ for } j=1,2,3\}$. All of these regions have the same shape with
$\Omega_{1,(1,1,1)}=\{(x_1,x_2,x_3)\in \mathbb{R}^3|x_1\geq x_2\geq0,x_1\geq x_3\geq 0 \}$, which means they can coincide with each other through rotation, shift and reflection.
So, we define $\Omega=\Omega_{1,(1,1,1)}$ and without loss of generality assume
$|Q\cap \Omega|=\max_{i\in[3], \mathbf{s}\in \{1,-1\}^3} |Q\cap \Omega_{i,\mathbf{s}}|$. Obviously, we have
$|Q\cap \Omega|\geq \frac{n}{24}$.

We only need to consider the case $|Q\cap\{p\}|<\frac{n}{24}$.
Introducing notations $\widetilde{\Omega}=\Omega \setminus \{p\}$, $\Omega^o= \Omega \setminus \partial \Omega$, from the property of $L_1$ median we know
$\sum_{q\in Q\setminus\{p\}}\frac{x_{q,1}-x_{p,1}}{\|q-p\|} \leq |Q \cap \{p\}|<\frac{n}{24}$. Since $p$ is the origin, we have$
\sum\nolimits_{q\in Q\cap \widetilde{\Omega}}\frac{x_{q,1}}{\|q\|} +\sum\nolimits_{q\in Q \setminus \Omega}\frac{x_{q,1}}{\|q\|}
<\frac{n}{24}$, which implies
$|Q\cap \widetilde{\Omega}|\frac{1}{\sqrt{3}}<\frac{n}{24}+
| Q \setminus \Omega|\leq \frac{n}{24}+(n-| Q \cap \widetilde{\Omega}|)
$
since $\frac{x_{q,1}}{\|q\|}\leq \frac{1}{\sqrt{3}}$, for all $ q\in \widetilde{\Omega} $.
Thus, we obtain
\begin{equation}  \label{size of Q cap widetilde{Omega} in R3}
|Q\cap \widetilde{\Omega}|<\frac{\sqrt{3}n}{1+\sqrt{3}}\cdot\frac{25}{24}<0.67n.
\end{equation}

Now, for $x=(x_1,x_2,x_3)\in \mathbb{R}^3$ we define $h_1(x)=x_1-x_2$, $h_2(x)=x_1-x_3$, $h_3(x)=x_2$, $h_4(x)=x_3$, and $H_i^+=\{x\in \mathbb{R}^3|\ h_i(x)\geq 0\}$,
$H_i^--=\{x\in \mathbb{R}^3|\ h_i(x)\leq 0\}$,
and assert there exists $i\in[4]$ such that
$|H_i^+\cap Q |\geq \frac{n}{24}$ and  $|H_i^- \cap Q |\geq \frac{n}{24}$.
Otherwise, since $|Q\cap \Omega|\geq \frac{n}{24}$ and $\Omega \subset \cap_{i=1}^4 H_i^+$,
we have  $|H_i^- \cap Q |<\frac{n}{24}$ for all $i\in [4]$.
From $\cup_{i=1}^4 H_i^- \cup \Omega^o=\mathbb{R}^3$ and \eqref{size of Q cap widetilde{Omega} in R3} we have
\begin{equation} \label{a contradiction}
\begin{split}
n=&|Q|=|\mathbb{R}^3\cap Q|=|(\cup_{i=1}^4H_i^- \cup \Omega^o)\cap Q|\leq
\sum_{i=1}^4|H_1^-\cap Q|+|\Omega^o\cap Q|\\
\leq&\sum_{i=1}^4|H_1^-\cap Q|+|\widetilde{\Omega}\cap Q|
\leq \frac{4n}{24}+0.67n<n,
\end{split}
\end{equation}
which is a contradiction.
Therefore,  in $\{x_1-x_2=0, x_1-x_3=0, x_2=0, x_3=0\}$ there exists at lease one plane such that any closed half space with this line as the boundary contains at least $\frac{n}{24}$ points of $Q$.

The second step is to obtain three sets of planes which have the same structure with $\Eu{H}_3$, and this can be done
through orthogonal transformation.
Since a plane through the origin can be uniquely determined by its normal vector, we can use normal vectors
$V_3=\{(1,0,0),(0,1,0),(0,0,1),(1,\pm1,0),(0,1,\pm1), (\pm1,0,1)\}$
to represent planes in $\Eu{H}_3$. Then, we choose three orthogonal matrices $M_1$, $M_2$, $M_3$ and define
$V_{3}(M_i)=\{vM_i|\ v\in V_3\}$ for $i=1,2,3$. One set of feasible orthogonal matrices is $\{M_i|\ M_i=I_3-2u_{3,i}u_{3,i}^T, \text{ for } i=1,2,3\}$, where $I_3$ is a $3\times 3$ identity matrix, and $u_{3,i}=(1^i,2^i,3^i)^T$ is a column vector.
It can be verified that $\min_{v_i\in V_{3}(M_i), \forall\ i\in[3]} |\textbf{Det}([v_1; v_2; v_3])|\geq  4.8468 \times 10^{-4}$,
where $[v_1; v_2; v_3]$ is a $3\times3$ matrix and $v_i$ is its $i^{\text{th}}$ row.
This means if we arbitrarily choose three vectors $v_1,v_2,v_3$ from $V_{3}(M_1)$, $V_{3}(M_2)$ and $V_{3}(M_3)$ respectively, then these three vectors are linearly independent, so the three planes determined by these vectors can form an oblique coordinates system. We can use the method in the proof of Theorem \ref{theorem the epsilon support for L_1 median in R^2}, to generate a set $T(v_1,v_2,v_3)$ with size $O(C_{[v_1;v_2;v_3]}\frac{k^3}{\eps^3})$ in this oblique coordinate system, where $C_{[v_1;v_2;v_3]}$ is a constant determined by $|\textbf{Det} ([ v 1 ; v 2 ; v 3 ])|$.
For the three orthogonal matrices we chose above, $|\textbf{Det} ([ v 1 ; v 2 ; v 3 ])|$ has a lower bound, so the
constant $C_{[v_1;v_2;v_3]}$ has an upper bound, which implies $O\left(C_{[v_1;v_2;v_3]}\frac{k^3}{\eps^3}\right)=O\left(\frac{k^3}{\eps^3}\right)$.

For any $L_1$ median $p$ of $n$ points $Q$ and any $V_{3}(M_i)$ there must be a plane through $p$ and orthogonal to a vector in $V_{3}(M_i)$ such that in both sides of this plane there are at least $\frac{n}{24}$ points of $Q$.
So, there exist $(v_1,v_2,v_3)\in V_{3}(M_1)\times V_{3}(M_2)\times V_{3}(M_3)$ and $x\in T(v_1,v_2,v_3)$  such that
$\|x-p\|\leq \eps \cost(p,Q)$. Therefore,  we can take $T=\cup_{v_i\in V_{3}(M_i), \forall\ i\in[3]} T(v_1,v_3,v_3)$ as an $\eps$-support for $L_1$ median on $\Eu{P}$
 with size $O(\frac{k^3}{\eps^3})$.
\end{proof}

In the proof of Theorem \ref{theorem the epsilon support for L_1 median in R^3},
we choose three orthogonal matrices $M_1$, $M_2$, $M_3$. These three matrices are independent from the input data $\Eu{P}$, so we can store these orthogonal matrices and use them to compute the $\eps$-support of $L_1$ median
for any $\Eu{P}$ in $\mathbb{R}^3$.

To generalize the result of Theorem \ref{theorem the epsilon support for L_1 median in R^3} to $\mathbb{R}^d$,
we can use $d+2 {n\choose 2} $ hyperplanes
$\Eu{H}_d=\{x_i=0 \mid i\in[d]\}\cup\{x_i\pm x_j=0 \mid \ 1\leq i< j\leq d\}$
 to partition
$\mathbb{R}^d$ into $d2^d$ regions: $\{\Omega_{i,\mathbf{s}} \mid i\in[d], \mathbf{s}\in\{1,-1\}^d\}$,
where $\Omega_{i,\mathbf{s}}=\Omega_{i,(s_1,\cdots,s_d)}:=\{(x_1,\cdots,x_d)\in \mathbb{R}^d \mid s_ix_i\geq s_jx_j\geq 0,\forall\  j\in[d]\}$. All of these regions have the same shape with
$\Omega_{1,(1,\cdots,1)}=\{(x_1,\cdots,x_d)\in \mathbb{R}^d \mid x_1\geq x_j\geq0,\text{for } j=2,\cdots,d \}$. Using the method in the proof of  Theorem \ref{theorem the epsilon support for L_1 median in R^3} we can show, if $p$ is the $L_1$ median of $n$ points $Q$ and is the origin, then there is a hyperplane $h$ in $\Eu{H}_d$ such that any half space with $h$ as the boundary contains at least $\frac{n}{d2^d}$ points of $Q$. (In $\mathbb{R}^d$, \eqref{a contradiction} will become $n\leq \frac{2(d-1)n}{d2^d}+\frac{\sqrt{d}n}{1+\sqrt{d}}\frac{d2^d+1}{d2^d}$, and it is easy to show the right side of this inequality is always less than $n$ for all $d\geq 3$, so the method in the proof of Theorem \ref{theorem the epsilon support for L_1 median in R^3} still works.)

Suppose $V_d$ is the collection of normal vectors of all hyperplanes in $\Eu{H}_d$.
We randomly choose a set of $d$-dimensional orthogonal matrices $\Eu{M}=\{M_1,\cdots,M_d\}$, and define $V_{d}(M_i)=\{vM_i \mid v\in V_d\}$ for $i=1,\cdots,d$. If $\min_{v_i\in V_{d}(M_i), \forall\ i\in[d]} \mid \textbf{Det}([v_1; \cdots; v_d])|\geq c_\Eu{M}>0$, where
$c_\Eu{M}$ is a positive constant dependent on $\Eu{M}$ and
$[v_1;\cdots; v_d]$ is a $d\times d$ matrix with $v_i$ as its $i$th row, then we can store these matrices, for each $(v_1,\cdots,v_d)\in V_{d}(M_1)\times \cdots\times V_{d}(M_d)$ build an oblique coordinate system, and use the method in Theorem \ref{theorem the epsilon support for L_1 median in R^2}, to generate a set $T(v_1,\cdots,v_d)$ with size $O(C_{[v_1;\cdots; v_d]}\frac{k^d}{\eps^d})=O(C_\Eu{M}\frac{k^d}{\eps^d})$, where $C_\Eu{M}$ is a positive constant dependent on $\Eu{M}$. Finally, we return
$T=\cup_{v_i\in V_{d}(M_i), \forall\ i\in[d]} T(v_1,\cdots,v_d)$ as an $\eps$-support for $L_1$ median on $\Eu{P}$,
and the size of $T$ is $|T|=O(C_\Eu{M}\frac{k^d}{\eps^d})=O(\frac{k^d}{\eps^d})$, since $\Eu{M}$ is fixed for all uncertain data in $\mathbb{R}^d$.

Since a $d$-dimensional orthogonal matrix has $d(d-1)/2$ independent variables,  we can always find
orthogonal matrices $M_1,\cdots,M_d$ and a constant $c_\Eu{M}$, such that
\[
\min_{v_i\in V_{d}(M_i), \forall\ i\in[d]}  |\textbf{Det}([v_1; \cdots; v_d])|\geq c_\Eu{M}>0,
\]
 and for fixed $d$,  $M_1,\cdots,M_d$ can be stored to deal with any input data $\Eu{P}$ in $\mathbb{R}^d$.
For example, for $d=4$ we can define $M_i=I_4-2u_{4,i}u_{4,i}^T$, for $i=1,\cdots,4$, where $I_4$ is an identity matrix and $u_{4,i}=(1^i,2^i,3^i, 4^i)^T$, and it can be verified that
$\min_{v_i\in V_{4}(M_i), \forall\ i\in[4]} |\textbf{Det}([v_1; \cdots; v_4])|\geq 3.7649\times 10^{-6}$.

For $d=5$, we can define $M_i=I_5-2u_{5,i}u_{5,i}^T$, for $i=1,\cdots,5$, where $I_5$ is an identity matrix and $u_{5,i}=(1^i,2^i,3^i, 4^i, 5^i)^T$, and we have
$\min_{v_i\in V_{5}(M_i), \forall\ i\in[5]} |\textbf{Det}([v_1; \cdots; v_5])| \geq 2.3635\times 10^{-11}$.
In summary, we have the following theorem.

\begin{theorem}
\label{theorem the epsilon support for L_1 median in R^d}
Given a set of $n$ uncertain points $\Eu{P}=\{P_1,\cdots,P_n\}$, where $P_i=\{p_{i,1},\cdots,p_{i,k}\}$ $\subset\mathbb{R}^d$,
 and $\eps \in(0,1]$, for and fixed $d$ we can construct an $\eps$-support $T$ for $L_1$ median on $\Eu{P}$ that has a size $|T|=O\left(\frac{k^d}{\eps^d}\right)$.
\end{theorem}

\newpage

\bibliography{p16-buchin}

\end{document}